\documentclass[a4paper,11pt]{amsart}
\usepackage[utf8]{inputenc}
\usepackage{amsaddr}
\usepackage{calrsfs}
\usepackage{graphicx,color}
\usepackage[english]{babel}
\usepackage{graphicx}
\usepackage{amsmath}
\usepackage{epstopdf}
\usepackage{float}
\usepackage{color}
\usepackage{cite}

\usepackage{amssymb}
\usepackage[final]{pdfpages}
\usepackage{amsmath}
\usepackage{amstext}
\usepackage{amssymb}
\usepackage{graphicx}
\usepackage[utf8]{inputenc}

\numberwithin{equation}{section}

\theoremstyle{definition}
\newtheorem{dfn}{Definition}[section]

\theoremstyle{plain}
\newtheorem{thm}{Theorem}[section]

\newtheorem{lmm}{Lemma}[section]
\theoremstyle{definition}

\newcommand{\e}{\mathrm{e}}

\begin{document}
\title[Fractional short rate ]
{Actuarial strategy for pricing Asian options under a mixed fractional Brownian motion with jumps}

\date{\today}

\author[Shokrollahi]{Foad Shokrollahi}
\address{Department of Mathematics and Statistics, University of Vaasa, P.O. Box 700, FIN-65101 Vaasa, FINLAND}
\email{foad.shokrollahi@uva.fi}

\author[Ahmadian]{Davood Ahmadian}
\address{ Faculty of Mathematical Sciences, University of Tabriz
29Blvd, Tabriz, Iran}
\email{d.ahmadian@tabrizu.ac.ir }

\author[Ballestra]{Luca Vincenzo Ballestra}
\address{Department of Statistical Sciences, Alma Mater Studiorum University of Bologna, Via delle Belle Arti 41, 40126 Bologna, Italy
}
\email{luca.ballestra@unibo.it}

\begin{abstract}
The mixed fractional Brownian motion ($mfBm$) has become quite popular in finance, since it allows one to model long-range dependence and self-similarity  while remaining, for certain values of the Hurst parameter, arbitrage-free. In the present paper, we propose approximate closed-form solutions for pricing arithmetic Asian options on an underlying described by the $mfBm$. Specifically, we consider both arithmetic Asian options and arithmetic Asian power options, and we obtain analytical formulas for pricing them based on a convenient approximation of the strike price. Both the standard $mfBm$ and the $mfBm$ with Poisson log-normally distributed jumps are taken into account.
\end{abstract}


\maketitle

\section{Introduction}
As documented by several empirical studies, returns of financial assets are often affected by long-range dependence and
self-similarity, and thus several scholars have proposed to model them by using the so-called fractional Brownian motion ($fBm$) with
Hurst exponent $H \in (1/2, 1)$, see, e.g., \cite{Ding}, \cite{Lo}, \cite{Rostek}. Nevertheless, as shown in \cite{BenderMOR}, and \cite{Cheridito2}, the $fBm$ is not arbutrage-free, and thus, some authors \cite{Cheridito}, \cite{MishuraBook}, \cite{Mishura}, \cite{Zili} have proposed to use the so-called mixed fractional Brownian motion ($mfBm$), i.e., a linear combination of the Wiener process and the $fBm$ itself.
When the Hurst exponent of the $fBm$ is greater than $1/2$, the $mfBm$ results in a Gaussian long-memory process suitable for modeling
the returns of financial assets. Furthermore, as shown in \cite{Cheridito}, if $H$ is greater than $3/4$ the $mfBm$ is arbitrage-free. Therefore, the $mfBm$ has been largely employed for pricing various kinds of derivatives, including stock options \cite{Ballestra}, \cite{Chen},
\cite{Wang2}, equity warrants \cite{Xiao}, currency options \cite{Sun,shokrollahi2015actuarial,shokrollahi2014pricing}, and credit derivatives \cite{He}.
\par
The $mfBm$ has also been used for pricing Asian options, i.e. options whose payoff depends on the average, either geometric or arithmetic, of the underlying asset over the whole option's lifespan \cite{shokrollahi2018evaluation}. In particular, \cite{Prakasa} obtained an exact closed-form solution for geometric Asian options, \cite{ZhangPower} derived an exact analytical formula for geometric Asian power options, \cite{Wang2} proposed an exact closed-form solution for geometric Asian rainbow options on two underlying assets, which was generalized to the case on an arbitrarily large number of underlying assets in
\cite{AhmadianPhysicaA}. Furthermore, \cite{Peng} considered a $mfBm$ with jumps and obtained an exact analytical formula for valuing geometric Asian power options.\par
All the exact closed-form solutions proposed in the aforementioned papers are valid for Asian options of geometric type. By contrast, for arithemtic Asian options, exact closed-form solutions cannot be found, and thus some numerical approximation is required. To this aim, some researchers have proposed to apply the Monte Carlo simulation method, which could also be enhanced by using the price of a geometric Asian option as a control variate (see \cite{Peng}). Nevertheless, this approach, albeit flexible and relatively simple to implement, has the disadvantage of being very slow to converge.
\par Therefore, in this paper we propose the use of approximate closed-form solutions. Specifically, we consider both arithmetic Asian options and arithmetic Asian power options, and we value them by applying an analytical formula that we obtain by a convenient approximation of the strike price. Both the standard $mfBm$ and the $mfBm$ with Poisson log-normally distributed jumps are taken into account.
\par
The remainder of the paper is organized as follows: in Section 2 we describe the $mfBm$ with jumps and present exact analytical formulas for pricing Asian options of geometric type. In Section 3 we derive an approximate closed-form solution for pricing arithmetic Asian options, which in Section 4 we extend to arithmetic Asian power options. Finally, Section 5 concludes.

\section{Preliminaries}
We assume that the following assumptions hold:
\begin{enumerate}
\item[(i)] the dynamics of the stock price is governed by the following equation
\begin{eqnarray}
S_t=S_0\exp\{f(t)+\sigma B_t+\varepsilon B_t^H-\frac{1}{2}\sigma^2t-\frac{1}{2}\varepsilon^2t^{2H}+\sum_{i=1}^{N_t}J_i\},
\label{stock price 1}
\end{eqnarray}

where $f(t)$ is a non-random function of $t$, and, $\sigma$, $\varepsilon$ are
constants; $B_t$ and $B_t^H$ are a standard Brownian motion and a $fBm$ respectively; $N_t$ is a poisson process with rate
$\lambda$; $e^{J_i}-1$ is jump size which is a sequence of independent identically distributed and $J_i\sim N(\mu_J, \sigma_J^2)$. Moreover, all three
sources of randomness, the $B_t$ and $B_t^H$ and  $N_t$ are
supposed to be independent;

\item[(ii)]  the risk free interest rate $r$ and dividend rate $q$ are known and constant
through time;

\item[(iii)] the option can be exercised only at the maturity time;
\item[(iv)] there are no transaction costs in buying or selling the stocks or option.
\end{enumerate}

Based on option valuation strategy and risk neutral measure (cf.Refs.\cite{cheridito2003arbitrage,shreve2004stochastic}), the dynamics of the stock price process $S(t)$ can be written as:

\begin{eqnarray}
S_t=S_0\exp[(r-q)t+\sigma B_t+\varepsilon B_t^H-\frac{1}{2}\sigma^2t-\frac{1}{2}\varepsilon^2t^{2H}+\sum_{i=1}^{N_t}J_i].
\label{stock price 2}
\end{eqnarray}

\begin{lmm}

\begin{eqnarray}
E[S_T]=S_0e^{(r-q)T+\lambda(\rho-1)T}.
\end{eqnarray}
\label{lemma1}
where $\rho=E[e^{J_1}]$.
\end{lmm}

\begin{proof}
\begin{eqnarray}
E[S_T]&=&S_0E[e^{(r-q)t+\sigma B_t+\varepsilon B_t^H-\frac{1}{2}\sigma^2t-\frac{1}{2}\varepsilon^2t^{2H}+\sum_{i=1}^{N_T}J_i}\nonumber\\
&=&S_0E[e^{(r-q)t+\sigma B_t+\varepsilon B_t^H-\frac{1}{2}\sigma^2t-\frac{1}{2}\varepsilon^2t^{2H}}]E[\prod_{i=1}^{N_T}e^{J_i}]\nonumber\\
&=&S_0e^{(r-q)T}E[\prod_{i=0}^{N_T}e^{J_i}],
\end{eqnarray}
and
\begin{eqnarray}
E\left[\prod_{i=0}^{N_T}e^{J_i}\right]&=&\sum_{n=0}^{\infty}\frac{(\lambda T)^ne^{-\lambda T}}{n!}E[\prod_{i=1}^{n}e^{J_i}]\nonumber\\
&=&\sum_{n=0}^{\infty}\frac{(\lambda T)^ne^{-\lambda T}}{n!}\rho^n\nonumber\\
&=&e^{\lambda(\rho-1)T}\sum_{n=0}^{\infty}\frac{(\lambda T)^ne^{-\rho\lambda T}}{n!}\nonumber\\
&=&e^{\lambda(\rho-1)T}.
\end{eqnarray}

Then,

\begin{eqnarray}
E[S_T]=S_0e^{(r-q)T+\lambda(\rho-1)T}.
\end{eqnarray}

\end{proof}

This section deals with  the new pricing model for the currency options using the actuarial approach, when the spot exchange rate follows the $MFBM$ with jumps process. This model can be applied to different financial markets such as: in the arbitrage-free, equilibrium and complete markets and also in the arbitrage,non-equilibrium and incomplete markets.

\begin{dfn}(\cite{bladt1998actuarial}) The expectation return rate $\theta(t)$  of $S_t$ on $t\in[0,T]$ is defined to $\int_0^T\theta(s)ds$  as follows:
\begin{eqnarray}
\frac{E\Big(S(T)\Big)}{S_0}=\exp\Big(\int_0^T\theta(s)ds\Big),
\end{eqnarray}
where $\theta(t)$ is the continuously compounding ram of return of $S_t$ at time $t$.
\label{actuarial}
\end{dfn}

\begin{dfn}  Let $A(T)$ be the average price of the underlying asset over the predetermined interval. Then, the value of geometric Asian options with exercise price $K$ and maturity time $T$ by actuarial approach are as follows:
\begin{eqnarray}
C(S_T,K)=E\Big[\Big(\exp\Big(-\int_0^T\theta(t)dt\Big)A(T)-Ke^{-rT}\Big)I_B\Big],
\label{eq:callprice}
\end{eqnarray}
for Asian call option, where $B=\exp\Big(-\int_0^T\theta(t)dt\Big)A(T)>Ke^{-rT}$.
\begin{eqnarray}
P((S_T,K)=E\Big[\Big(Ke^{-rT}-\exp\Big(-\int_0^T\theta(t)dt\Big)A(T)\Big)I_{B'}\Big],
\label{eq:putprice}
\end{eqnarray}
for Asian put option, $B'=\exp\Big(-\int_0^T\theta(t)dt\Big)A(T)<Ke^{-rT}$.
\label{Actuarial price}
\end{dfn}

\begin{thm}

Let the stock price $S_T$ satisfy Equation (\ref{stock price 2}). Then, the value of geometric Asian options with exercise price $K$ is given by
\begin{eqnarray}
C(S_T,K)&=&\nonumber\\
&=&\sum_{n=0}^{\infty}\frac{(\lambda T)^ne^{-\lambda T}}{n!}\left[e^{-(r-q)T-\lambda(\rho-1)T+\hat{\mu}+\frac{1}{2}\hat{\sigma}^2}\Phi(d_1)-Ke^{-rT}\Phi(d_2)\right],
\label{geometricjump}
\end{eqnarray}
for the call options,
and
\begin{eqnarray}
P(S_T,K)&=&\nonumber\\
&=&\sum_{n=0}^{\infty}\frac{(\lambda T)^ne^{-\lambda T}}{n!}\left[Ke^{-rT}\Phi(-d_2)-\e^{-(r-q)T-\lambda(\rho-1)T+\hat{\mu}+\frac{1}{2}\hat{\sigma}^2}\Phi(-d_1)\right],
\end{eqnarray}
for the put options, where
\begin{eqnarray}
&&d_2=\frac{\hat{\mu}-\ln U}{\hat{\sigma}}, \, d_1=d_2+\hat{\sigma}, \, U=Ke^{-qT-(1-\rho)\lambda T}, \nonumber\\
&&\rho=E(e^J)=e^{\mu_J+\frac{1}{2}\sigma_J^2},\,\hat{\mu}=\ln S_0+\frac{1}{2}(r_n-\frac{1}{2}\sigma_n^2)T, \, \hat{\sigma}^2=\sigma_n^2\frac{T}{3},\nonumber\\
&&r_n=r-q+\frac{n}{T}(\mu_J+\frac{1}{2}\sigma_J^2), \, \sigma_n^2=\sigma^2+\frac{1}{T}(\varepsilon^2 T^{2H}+n\sigma_J^2).
\end{eqnarray}
and $\Phi(.)$ is cumulative normal density function.
\label{TheoremGeometric}
\end{thm}

\begin{proof}
We know that $B_T\sim N(0, T)$ and  $B_T^H\sim N(0, T^{2H})$. If $Z_1\sim N(0, 1)$ and $Z_2\sim N(0, 1)$, then $B_T=\sqrt{T}Z_1$ and $B_T^H=\sqrt{T^{2H}}Z_2$. Given that $J_i\sim N(\mu_J, \sigma_J^2)$ and $e^{J_i}-1$ is a sequence of independent identically distributed, we have $\sum_{i=1}^{n}J_i \sim N\left(n\mu_J, n\sigma_J^2\right)$; and for
$Z_3\sim N(0, 1)$ we have  $\sum_{i=1}^{n}J_i=n\mu_J+\sqrt{n}\sigma_JZ_3$. Then for $Z_n\sim (0, 1)$, it follows

\begin{eqnarray}
\sigma_n\sqrt{T}Z_n=\sigma\sqrt{T}Z_1+\varepsilon\sqrt{T^{2H}}Z_2+\sqrt{n}\sigma_JZ_3,
\label{eq:n}
\end{eqnarray}

where $\sigma_n^2=\sigma^2+\frac{1}{T}\left(\varepsilon^2T^{2H}+n\sigma_J^2\right)$. Consequently, assuming $N_T=n$, from Eqs. (\ref{stock price 2}) and (\ref{eq:n}) we obtain

\begin{eqnarray}
S_T^n&=&S_0\exp\left[(r-q)T-\frac{1}{2}\sigma^2T-\frac{1}{2}\varepsilon^2T^{2H}+n\mu_J+\sigma_n\sqrt{T}Z_n\right]\nonumber\\
&=&S_0\exp\left[(r_n-\frac{1}{2}\sigma_n^2)T+\sigma_n\sqrt{T}Z_n\right],
\label{eq:Sn}
\end{eqnarray}
here
\begin{eqnarray}
r_n=r-q+\frac{n}{T}(\mu_J+\frac{1}{2}\sigma_J^2).
\end{eqnarray}

Let,

\begin{eqnarray}
L(T)=\frac{1}{T}\int_0^T\ln S_tdt,
\end{eqnarray}
and
\begin{eqnarray}
G(T)=\exp\left\{L(T)\right\}.
\end{eqnarray}

We know that the random variable $L(T)$ has Gaussian distribution under the risk-neutral probability measure. Let $\hat{\mu}$ and $\hat{\sigma}^2$ denote the mean and the variance of the random variable $L(T)$, respectively. If $N_T=n$, then

\begin{eqnarray}
\hat{\mu}&=&E[L(T)]=\frac{1}{T}\int_0^TE\left[\ln S_t^n\right]dt=\ln S_0+\frac{1}{T}\int_0^T(r_n-\frac{1}{2}\sigma_n^2)tdt\nonumber\\
&=&\ln S_0+\frac{1}{2}(r_n-\frac{1}{2}\sigma_n^2)T,
\end{eqnarray}

and

\begin{eqnarray}
\hat{\sigma}^2&=&Var[L(T)]=E[L(T)-\hat{\mu}]^2\nonumber\\
&=&\frac{1}{T^2}\sigma_n^2\int_0^T\int_0^TE\left(\sqrt{t}Z_n(t)\sqrt{s}Z_n(s)\right)dtds\nonumber\\
&=&\frac{1}{T^2}\sigma_n^2\int_0^T\int_0^T[|t|+|s|-|t-s|]dtds\nonumber\\
&=&\sigma_n^2\frac{T}{3}
\end{eqnarray}

The random variable $G(T)$ is log-normally distributed and the random variable $\ln G(T)$ has the Gaussian distribution with the mean $\hat{\mu}$ and the variance $\hat{\sigma}^2$ as obtained above. Then, from Eq. (\ref{eq:callprice}), the price of geometric Asian call option is:

\begin{eqnarray}
C(S_T,K)&=&E\left[\left(\exp\left(-\int_0^T\theta(t)dt\right)G(T)-Ke^{-rT}\right)I_B\right]\nonumber\\
&=&E\left[\left(e^{-(r-q)T-\lambda(\rho-1)T}G(T)-Ke^{-rT}\right)I_B\right]\nonumber\\
&=&e^{-(r-q)T-\lambda(\rho-1)T}E\left[G(T)I_B\right]-Ke^{-rT}E\left[I_B\right]\nonumber\\
&=&e^{-(r-q)T-\lambda(\rho-1)T}I_1-Ke^{-rT}I_2,
\label{I}
\end{eqnarray}
where $I_1=E\left[G(T)I_B\right]$ and $I_2=E\left[I_B\right]$.
Let

\begin{eqnarray}
B&=&\left\{x:e^x>U=Ke^{-qT-(1-\rho)\lambda T}\right\}\nonumber\\
&=&\left\{x:G(T)>U\right\}\nonumber\\
&=&\left\{y:e^{\hat{\mu}+\hat{\sigma}y}>U\right\}=\left\{y:\hat{\mu}+\hat{\sigma}y>\ln U\right\}\nonumber\\
&=&\left\{y:y>-d_2=\frac{\ln U-\hat{\mu}}{\hat{\sigma}}\right\}.
\end{eqnarray}
Observe that

\begin{eqnarray}
I_2&=&E\left[I_B\right]=\sum_{n=0}^{\infty}E\left[I_{G(T)>-d_2}\right]P\left[N_T=n\right]\nonumber\\
&=&\sum_{n=0}^{\infty}\frac{(\lambda T)^ne^{-\lambda T}}{n!}E\left[I_{G(T)>-d_2}\right]\nonumber\\
&=&\sum_{n=0}^{\infty}\frac{(\lambda T)^ne^{-\lambda T}}{n!}\int_{-d_2}^{\infty}\phi(y)dy\nonumber\\
&=&\sum_{n=0}^{\infty}\frac{(\lambda T)^ne^{-\lambda T}}{n!}\Phi(d_2),
\label{I2}
\end{eqnarray}

and

\begin{eqnarray}
I_1&=&E\left[G(T)I_B\right]=\sum_{n=0}^{\infty}E\left[G(T)I_{G(T)>-d_2}\right]P\left[N_T=n\right]\nonumber\\
&=&\sum_{n=0}^{\infty}\frac{(\lambda T)^ne^{-\lambda T}}{n!}E\left[e^{L(T)}I_{G(T)>-d_2}\right]\nonumber\\
&=&\sum_{n=0}^{\infty}\frac{(\lambda T)^ne^{-\lambda T}}{n!}\int_{B}e^x\frac{1}{\sqrt{2\pi}\hat{\sigma}}e^{-\frac{(x-\hat{\mu})^2}{2\hat{\sigma}^2}}dx\nonumber\\
&=&\sum_{n=0}^{\infty}\frac{(\lambda T)^ne^{-\lambda T}}{n!}\int_{B}e^{\hat{\mu}+\hat{\sigma}y}\frac{1}{\sqrt{2\pi}\hat{\sigma}}e^{-\frac{(x-\hat{\mu})^2}{2\hat{\sigma}^2}}dy\nonumber\\
&=&\sum_{n=0}^{\infty}\frac{(\lambda T)^ne^{-\lambda T}}{n!}e^{\hat{\mu}+\frac{1}{2}\hat{\sigma}^2}\int_{-d_2}^{\infty}\frac{1}{\sqrt{2\pi}}e^{-\frac{1}{2}(y-\hat{\sigma})^2}dy\nonumber\\
&=&\sum_{n=0}^{\infty}\frac{(\lambda T)^ne^{-\lambda T}}{n!}e^{\hat{\mu}+\frac{1}{2}\hat{\sigma}^2}\int_{-d_2-\hat{\sigma}}^{\infty}\phi(y)dy\nonumber\\
&=&\sum_{n=0}^{\infty}\frac{(\lambda T)^ne^{-\lambda T}}{n!}e^{\hat{\mu}+\frac{1}{2}\hat{\sigma}^2}\Phi(d_1=d_2+\hat{\sigma}).
\label{I1}
\end{eqnarray}

From Eqs. (\ref{I1}), (\ref{I2}) and (\ref{I}), we get desired result.
\end{proof}

\begin{figure}[!htp]
\begin{center}
\resizebox*{13.5cm}{!}{\includegraphics{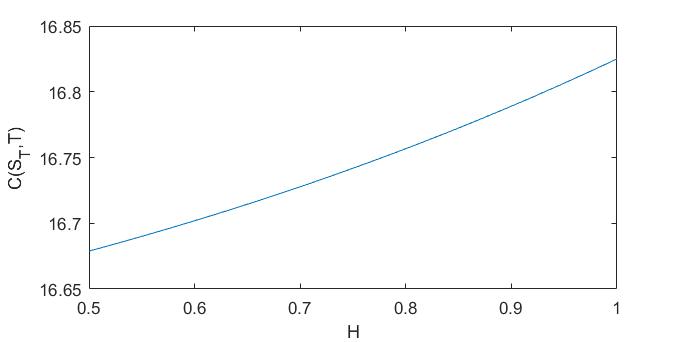}}\hspace{5pt}
\caption{The variation of the geometric Asian option with respect to the Husrt parameter $H$.  }\label{fig1}
\end{center}
\end{figure}
\begin{figure}[H]
\begin{center}
\resizebox*{13.5cm}{!}{\includegraphics{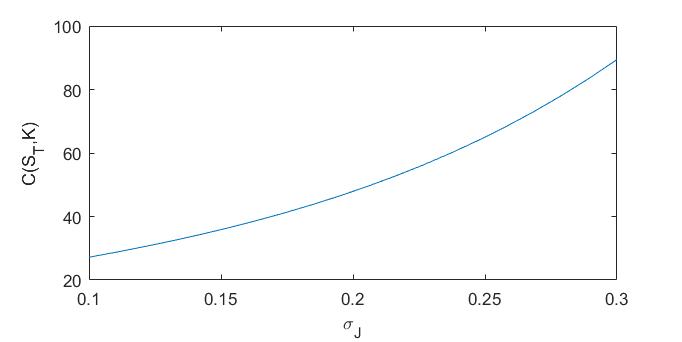}}\hspace{5pt}
\caption{The variation of the geometric Asian option with respect to the  $\sigma_J$.  }\label{fig2}
\end{center}
\end{figure}
\begin{figure}[H]
\begin{center}
\resizebox*{13.5cm}{!}{\includegraphics{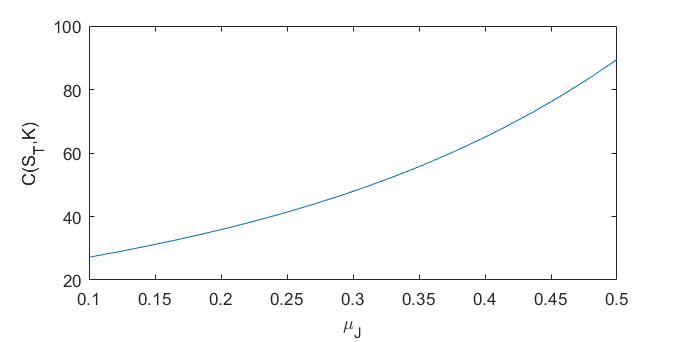}}\hspace{5pt}
\caption{The variation of the geometric Asian option with respect to the  $\mu_J$.  }\label{fig3}
\end{center}
\end{figure}

In Figs \ref{fig1}-\ref{fig3} we show the price of the geometric Asian option computed using formula (\ref{geometricjump}) as a function of the Hurst and the jump parameters.
\begin{figure}[H]
\begin{center}
\resizebox*{13.5cm}{!}{\includegraphics{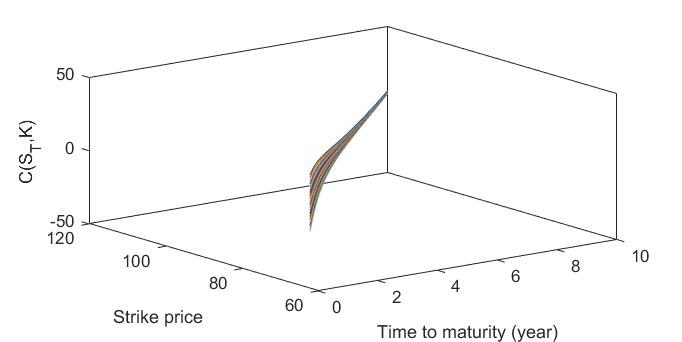}}\hspace{5pt}
\caption{The variation of the geometric Asian option with respect to the  strike price $K$ and time to maturity (year) $T$. }\label{fig4}
\end{center}
\end{figure}

In Fig \ref{fig4} we show the price of the geometric Asian option computed using formula (\ref{geometricjump}) as a function of the strike price and of the maturity.
\section{Arithmetic Asian option}

\begin{lmm} (\cite{vorst1992prices})

An approximate value of the arithmetic Asian option is given by

\begin{eqnarray}
\tilde{C}(S_T,K)&=&e^{-rT}E\left[(G(T)-K')^+\right],
\end{eqnarray}
where $G(T)$ ia the geometric Asian option and $K'$ is the adjusted strike price as follows

\begin{eqnarray}
K'=K-\left(E[A]-E[G]\right).
\end{eqnarray}
\label{Vorst}
\end{lmm}

Now, based on the actuarial approach an approximate of the arithmetic Asian options are proposed \cite{turnbull1991quick,kemna1990pricing}.

\begin{thm}

An approximate value of the arithmetic Asian option is given by

\begin{eqnarray}
\tilde{C}(S_T,K)&=&\nonumber\\
&=&\sum_{n=0}^{\infty}\frac{(\lambda T)^ne^{-\lambda T}}{n!}\left[e^{-(r-q)T-\lambda(\rho-1)T+\hat{\mu}+\frac{1}{2}\hat{\sigma}^2}\Phi(d_1)-K'e^{-rT}\Phi(d_2)\right],
\end{eqnarray}
for the call options,
and
\begin{eqnarray}
\tilde{P}(S_T,K)&=&\nonumber\\
&=&\sum_{n=0}^{\infty}\frac{(\lambda T)^ne^{-\lambda T}}{n!}\left[K'e^{-rT}\Phi(-d_2)-\e^{-(r-q)T-\lambda(\rho-1)T+\hat{\mu}+\frac{1}{2}\hat{\sigma}^2}\Phi(-d_1)\right],
\end{eqnarray}
for the put options, where
\begin{eqnarray}
&&d_2=\frac{\hat{\mu}-\ln U}{\hat{\sigma}}, \, d_1=d_2+\hat{\sigma},\nonumber\\
&&\rho=E(e^J)=e^{\mu_J+\frac{1}{2}\sigma_J^2},\, U=K'e^{-qT-(1-\rho)\lambda T},\nonumber\\
&&\hat{\mu}=\ln S_0+\frac{1}{2}(r_n-\frac{1}{2}\sigma_n^2)T,\,\bar{\mu}=\frac{S_0}{r_nT}\left(e^{r_nT}-1\right)\nonumber\\
&& \hat{\sigma}^2=\sigma_n^2\frac{T}{3},\, \sigma_n^2=\sigma^2+\frac{1}{T}(\varepsilon^2 T^{2H}+n\sigma_J^2)\nonumber\\
&&K'=K+E[G]-\bar{\mu},\,r_n=r-q+\frac{n}{T}(\mu_J+\frac{1}{2}\sigma_J^2),~E[G]=e^{\hat{\mu}+\frac{1}{2}\hat{\sigma}^2}.
\end{eqnarray}
and $\Phi(.)$ is cumulative normal density function.
\label{TheoremArithmetic}
\end{thm}

\begin{proof}
Suppose
\begin{eqnarray}
A(T)=\frac{1}{T}\int_0^TS_tdt
\end{eqnarray}
be an arithmetic average with mean $\bar{\mu}=E\left[A(T)\right]$, then from Lemma \ref{lemma1} we can get

\begin{eqnarray}
\bar{\mu}&=&E\left[A(T)\right]=\frac{1}{T}\int_0^TE\left[S_T\right]dt\nonumber\\
&=&\frac{S_0}{T}\int_0^Te^{{r_nt}}dt\nonumber\\
&=&\frac{S_0}{r_nT}\left(e^{r_nT}-1\right).
\end{eqnarray}

From Theorem \ref{Vorst}, to provide an approximation of the arithmetic Asian options, we need to calculate the new exercise price $K'$,

\begin{eqnarray}
K'&=&K-\left(E[A]-E[G]\right)\nonumber\\
&=&K+E[G]-\bar{\mu},
\end{eqnarray}
where $E[G]=e^{\hat{\mu}+\frac{1}{2}\hat{\sigma}^2}$ is calculated in Theorem \ref{TheoremGeometric}. Now, by replacing $K$ with $K'$ in the proof of Theorem \ref{TheoremGeometric}, the proof is completed.
\end{proof}

Suppose, $G$ and $A$ be the geometric and arithmetic average, respectively. Since the geometric average is less than the arithmetic average, $G<A$, the geometric option price is a lower bound for the arithmetic option price

\begin{eqnarray}
C_G=e^{-rT}E\left[(G-K)^+\right]\leq e^{-rT}E\left[(A-K)^+\right]=C_A.
\end{eqnarray}
The inequality

\begin{eqnarray}
\max (A-K, 0)&=&\max(G-K, G-A)+A-G\nonumber\\
&\leq& \max(G-K, 0)+A-G,
\end{eqnarray}

leads to the upper bound

\begin{eqnarray}
&&e^{-rT}E\left[\max(G-K, 0)+A-G\right]\nonumber\\
&&=e^{-rT}\left\{E\left[\max(G-K)^+\right]+E[A]-E[G]\right\}\nonumber\\
&&=C_G+e^{-rT}\left(E[A]-E[G]\right).
\end{eqnarray}
Therefore, we have the following bounds on the arithmetic option price

\begin{eqnarray}
C_G\leq C_A\leq C_G+e^{-rT}\left(E[A]-E[G]\right).
\end{eqnarray}

To compute the upper bound we need to compute the mean of arithmetic average, which is obtained in Theorem \ref{TheoremArithmetic}.

\begin{thm}

The pricing error of arithmetic approximate vale $\tilde{C}$ is given by

\begin{eqnarray}
|\tilde{C}-C_A|&=&e^{-rT}\left(E[A]-E[G]\right).
\end{eqnarray}

\label{Vorst1}
\end{thm}

\begin{proof}
Sine
\begin{eqnarray}
\max(G-K,0)&\leq&\max\left(G-K+E[A]-E[G], 0\right)\nonumber\\
&=&\max\left(G-K+E[G]-E[A]\right)+E[A]-E[G]\nonumber\\
&\leq&\max\left(G-K, 0\right)+E[A]-E[G],
\end{eqnarray}
it follows that
\begin{eqnarray}
C_G&=&e^{-rT}E\left[(G-K)^+\right]\nonumber\\
&\leq&e^{-rT}E\left[(G-K+E[A]-E[G])^+\right]\nonumber\\
&=&\tilde{C}\nonumber\\
&\leq&e^{-rT}E\left[(G-K)^++E[A]-E[G]\right]\nonumber\\
&=&C_G+e^{-rT}\left(E[A]-E[G]\right).
\end{eqnarray}
The analogous bounds exert for the exact value $C_A$, then the absolute difference between $\tilde{C}$ and $C_A$ is less than the difference between the upper and lower bounds
\begin{eqnarray}
|\tilde{C}-C_A|&\leq&C_G+e^{-rT}\left(E[A]-E[G]\right)-C_G\nonumber\\
&=&e^{-rT}\left(E[A]-E[G]\right).
\end{eqnarray}
\end{proof}

\section{Asian power options}

This section is deals with evaluate Asian power options by using actuarial approach in a mixed fractional Brownian motion with jumps environment.
From [some references], we know that the payoff function of Asian options is given by $(G^m(T)-K)^+$ and $(A^m(T)-K)^+$ for the geometric and arithmetic call options with some integer $m\geq1$, respectively.

\begin{thm}

Let the stock price $S(T)$ satisfy Equation (\ref{stock price 2}). Then, the value of geometric Asian power options with exercise price $K$ is given by
\begin{eqnarray}
C(S_T,K)&=&\nonumber\\
&=&\sum_{n=0}^{\infty}\frac{(\lambda T)^ne^{-\lambda T}}{n!}\left[e^{-(r-q)T-\lambda(\rho-1)T+m\hat{\mu}+\frac{1}{2}m^2\hat{\sigma}^2}\Phi(d_1)-Ke^{-rT}\Phi(d_2)\right],
\label{powerjump}
\end{eqnarray}
for the call options,
and
\begin{eqnarray}
P(S_T,K)&=&\nonumber\\
&=&\sum_{n=0}^{\infty}\frac{(\lambda T)^ne^{-\lambda T}}{n!}\left[Ke^{-rT}\Phi(-d_2)-\e^{-(r-q)T-\lambda(\rho-1)T+m\hat{\mu}+\frac{1}{2}m^2\hat{\sigma}^2}\Phi(-d_1)\right],
\end{eqnarray}
for the put options, where $d_2=\frac{\hat{\mu}-\ln U^{\frac{1}{n}}}{\hat{\sigma}}$ and other parameters are defined in Theorem \ref{TheoremGeometric}.
\label{TheoremPowerGeometric}
\end{thm}

\begin{proof}
For the geometric Asian power options, the payoff function is $(G^m(T)-K)^+=(e^{mL(T)}-K)^+$. Using the strategy given in proof of Theorem \ref{TheoremGeometric}, then we have

\begin{eqnarray}
C(S_T,K)&=&E\left[\left(\exp\left(-\int_0^T\theta(t)dt\right)G^m(T)-Ke^{-rT}\right)I_B\right]\nonumber\\
&=&E\left[\left(e^{-(r-q)T-\lambda(\rho-1)T}G^m(T)-Ke^{-rT}\right)I_B\right]\nonumber\\
&=&e^{-(r-q)T-\lambda(\rho-1)T}E\left[G^m(T)I_B\right]-Ke^{-rT}E\left[I_B\right]\nonumber\\
&=&e^{-(r-q)T-\lambda(\rho-1)T}I_1-Ke^{-rT}I_2,
\label{I111}
\end{eqnarray}
where $I_1=E\left[A(T)I_B\right]$ and $I_2=E\left[I_B\right]$.
Let

\begin{eqnarray}
B&=&\left\{x:e^{mx}>U=Ke^{-qT-(1-\rho)\lambda T}\right\}\nonumber\\
&=&\left\{x:G^m(T)>U\right\}\nonumber\\
&=&\left\{y:e^{m(\hat{\mu}+\hat{\sigma}y)}>U\right\}=\left\{y:m(\hat{\mu}+\hat{\sigma}y)>\ln U\right\}\nonumber\\
&=&\left\{y:y>-d_2=\frac{\ln U^{\frac{1}{m}}-\hat{\mu}}{\hat{\sigma}}\right\}.
\end{eqnarray}
Observe that

\begin{eqnarray}
I_2&=&E\left[I_B\right]=\sum_{n=0}^{\infty}E\left[I_{G^m(T)>-d_2}\right]P\left[N_T=n\right]\nonumber\\
&=&\sum_{n=0}^{\infty}\frac{(\lambda T)^ne^{-\lambda T}}{n!}E\left[I_{G^m(T)>-d_2}\right]\nonumber\\
&=&\sum_{n=0}^{\infty}\frac{(\lambda T)^ne^{-\lambda T}}{n!}\int_{-d_2}^{\infty}\phi(y)dy\nonumber\\
&=&\sum_{n=0}^{\infty}\frac{(\lambda T)^ne^{-\lambda T}}{n!}\Phi(d_2),
\label{I112}
\end{eqnarray}

and

\begin{eqnarray}
I_1&=&E\left[G^m(T)I_B\right]=\sum_{n=0}^{\infty}E\left[G^m(T)I_{G^m(T)>-d_2}\right]P\left[N_T=n\right]\nonumber\\
&=&\sum_{n=0}^{\infty}\frac{(\lambda T)^ne^{-\lambda T}}{n!}E\left[e^{mL(T)}I_{G^m(T)>-d_2}\right]\nonumber\\
&=&\sum_{n=0}^{\infty}\frac{(\lambda T)^ne^{-\lambda T}}{n!}\int_{B}e^{mx}\frac{1}{\sqrt{2\pi}\hat{\sigma}}e^{-\frac{(x-\hat{\mu})^2}{2\hat{\sigma}^2}}dx\nonumber\\
&=&\sum_{n=0}^{\infty}\frac{(\lambda T)^ne^{-\lambda T}}{n!}\int_{B}e^{m(\hat{\mu}+\hat{\sigma}y)}\frac{1}{\sqrt{2\pi}\hat{\sigma}}e^{-\frac{(x-\hat{\mu})^2}{2\hat{\sigma}^2}}dy\nonumber\\
&=&\sum_{n=0}^{\infty}\frac{(\lambda T)^ne^{-\lambda T}}{n!}e^{m\hat{\mu}+\frac{1}{2}m^2\hat{\sigma}^2}\int_{-d_2}^{\infty}\frac{1}{\sqrt{2\pi}}e^{-\frac{1}{2}(y-m\hat{\sigma})^2}dy\nonumber\\
&=&\sum_{n=0}^{\infty}\frac{(\lambda T)^ne^{-\lambda T}}{n!}e^{m\hat{\mu}+\frac{1}{2}m^2\hat{\sigma}^2}\int_{-d_2-\hat{\sigma}}^{\infty}\phi(y)dy\nonumber\\
&=&\sum_{n=0}^{\infty}\frac{(\lambda T)^ne^{-\lambda T}}{n!}e^{m\hat{\mu}+\frac{1}{2}m^2\hat{\sigma}^2}\Phi(d_1=d_2+\hat{\sigma}).
\label{I113}
\end{eqnarray}

From Eqs. (\ref{I111}), (\ref{I112}) and (\ref{I113}), we get desired result.

\end{proof}

In Figs \ref{fig5}-\ref{fig7} we show the price of the power Asian option computed using formula \ref{powerjump} as a function of the Hurst and the jump. parameters.
\begin{figure}[H]
\begin{center}
\resizebox*{13.5cm}{!}{\includegraphics{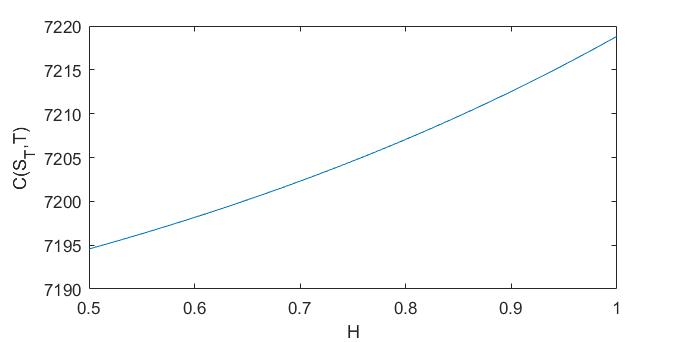}}\hspace{5pt}
\caption{The variation of the power Asian option with respect to the Husrt parameter $H$. }\label{fig5}
\end{center}
\end{figure}
\begin{figure}[H]
\begin{center}
\resizebox*{13.5cm}{!}{\includegraphics{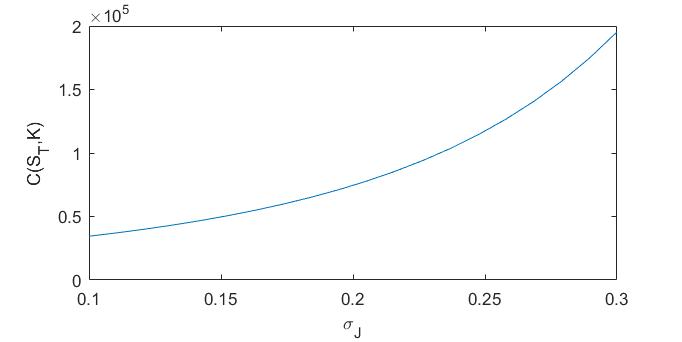}}\hspace{5pt}
\caption{The variation of the power Asian option with respect to the  $\sigma_J$}.\label{fig6}
\end{center}
\end{figure}
\begin{figure}[H]
\begin{center}
\resizebox*{13.5cm}{!}{\includegraphics{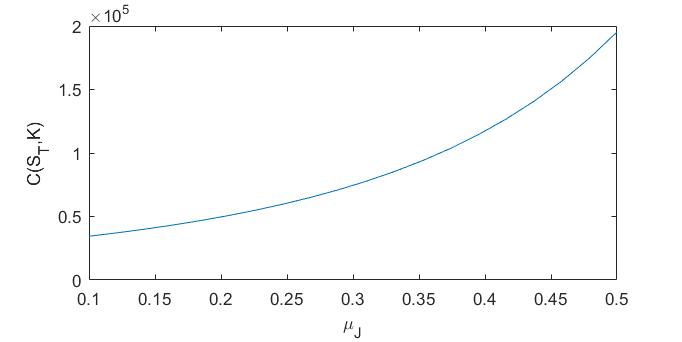}}\hspace{5pt}
\caption{The variation of the power Asian option with respect to the  $\mu_J$}.\label{fig7}
\end{center}
\end{figure}
\begin{figure}[H]
\begin{center}
\resizebox*{13.5cm}{!}{\includegraphics{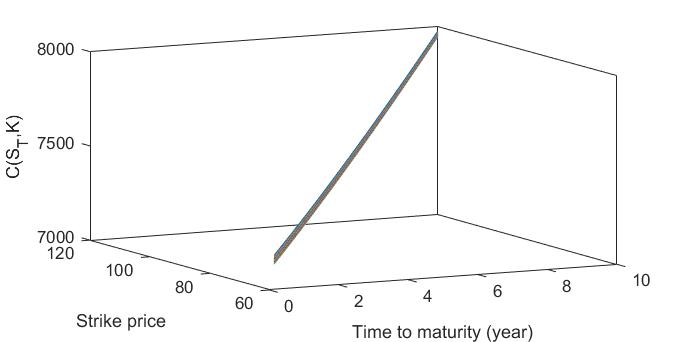}}\hspace{5pt}
\caption{The variation of the power Asian option with respect to the  both strike price $K$ and time to maturity (year) $T$.}\label{fig8}
\end{center}
\end{figure}

In Fig \ref{fig8} we show the price of the power Asian option computed using formula \ref{powerjump} as a function of the strike price and of the maturity.

The payoff function of the power Asian power call option is $(A^m(T)-K)$, then its approximate value is given by the following theorem.
\begin{thm}

An approximate value of the arithmetic Asian power option is given by

\begin{eqnarray}
\tilde{C}(S_T,K)&=&\nonumber\\
&=&\sum_{n=0}^{\infty}\frac{(\lambda T)^ne^{-\lambda T}}{n!}\left[e^{-(r-q)T-\lambda(\rho-1)T+\hat{\mu}+\frac{1}{2}\hat{\sigma}^2}\Phi(d_1)-K'e^{-rT}\Phi(d_2)\right],
\end{eqnarray}
for the call options,
and
\begin{eqnarray}
\tilde{P}(S_T,K)&=&\nonumber\\
&=&\sum_{n=0}^{\infty}\frac{(\lambda T)^ne^{-\lambda T}}{n!}\left[K'e^{-rT}\Phi(-d_2)-\e^{-(r-q)T-\lambda(\rho-1)T+\hat{\mu}+\frac{1}{2}\hat{\sigma}^2}\Phi(-d_1)\right],
\end{eqnarray}
for the put options, where $d_2=\frac{\hat{\mu}-\ln U^{\frac{1}{m}}}{\hat{\sigma}}$ and other parameters are defined in Theorem \ref{TheoremArithmetic}.
\label{TheoremPowerArithmetic}
\end{thm}

\begin{proof}
Suppose $\bar{\mu}=E\left[A^m(T)\right]$, then

\begin{eqnarray}
\bar{\mu}&=&E\left[A^m(T)\right]=\frac{1}{T}\int_0^TE\left[(S_t^n)^m\right]dt\nonumber\\
&=&\frac{S_0}{T}\int_0^Te^{m((r_n-\frac{1}{2}\sigma_n^2)t+\frac{1}{2}m\sigma_n^2t)}=\frac{S_0^m}{T}\int_0^Te^{mr_nt}dt\nonumber\\
&=&\frac{S_0^m}{mr_nT}\left(e^{mr_nT}-1\right).
\end{eqnarray}

From Theorem \cite{vorst1992prices}, to provide an approximation of the arithmetic Asian power options, we need to calculate the new exercise price $K'$,

\begin{eqnarray}
K'&=&K-\left(E[A^m]-E[G^m]\right)\nonumber\\
&=&K+E[G^m]-\bar{\mu},
\end{eqnarray}
where $E[G^m]=e^{m\hat{\mu}+\frac{1}{2}m^2\hat{\sigma}^2}$ is calculated in Theorem \ref{TheoremPowerGeometric}. Now, by replacing $K$ with $K'$ in the proof of Theorem \ref{TheoremPowerGeometric}, the proof is completed.
\end{proof}

\section{Conclusions}

We derived approximate closed-form solutions for pricing arithmetic Asian options on an underlying described by the mixed fractional Brownian motion ($mfBm$). In the last decade, the $mfBm$ has gained increasing popularity in finance, since it allows one to incorporate long-range dependence and self-similarity and, for certain values of the Hurst parameter, it is arbitrage-free.
In the present paper, suitable approximate closed-form solutions are obtained for both arithmetic Asian options and arithmetic Asian power options, by applying a convenient approximation of their  strike prices. Both the standard $mfBm$ and the $mfBm$ with Poisson log-normally distributed jumps are considered.

\bibliographystyle{plain}
\bibliography{reference}

\begin{thebibliography}{10}

\bibitem{AhmadianPhysicaA}
D.~Ahmadian and L.~V. Ballestra.
\newblock Pricing geometric {A}sian rainbow options under the mixed fractional
  {B}rownian motion.
\newblock {\em Physica {A}: {S}tatistical {M}echanics and its {A}pplications},
  {\bf555}:124458, 2020.

\bibitem{Ballestra}
L.~V. Ballestra, G.~Pacelli, and D.~Radi.
\newblock A very efficient approach for pricing barrier options on an
  underlying described by the mixed fractional {B}rownian motion.
\newblock {\em Chaos, {S}olitons \& {F}ractals}, {\bf87}:240--248, 2016.

\bibitem{BenderMOR}
C.~Bender and R.~J. Elliott.
\newblock Arbitrage in a discrete version of the {W}ick-fractional
  {B}lack-{S}choles market.
\newblock {\em Mathematics of {O}perations {R}esearch}, {\bf29}:935--945, 2004.

\bibitem{bladt1998actuarial}
Mogens Bladt and Tina~Hviid Rydberg.
\newblock An actuarial approach to option pricing under the physical measure
  and without market assumptions.
\newblock {\em Insurance: {M}athematics and {E}conomics}, {\bf22}(1):65--73,
  1998.

\bibitem{Chen}
J.~H Chen, F.~Y. Ren, and W.~Y. Qiu.
\newblock Option pricing of a mixed fractional-fractional version of the
  {B}lack-{S}choles model.
\newblock {\em Chaos, {S}olitons \& {F}ractals}, {\bf21}:1163--1174, 2004.

\bibitem{Cheridito}
P.~Cheridito.
\newblock Mixed fractional {B}rownian motion.
\newblock {\em Bernoulli}, {\bf7}:913--934, 2001.

\bibitem{Cheridito2}
P.~Cheridito.
\newblock Arbitrage in fractional {B}rownian motion models.
\newblock {\em Finance and {S}tochastics}, {\bf7}:533–553, 2003.

\bibitem{cheridito2003arbitrage}
Patrick Cheridito.
\newblock Arbitrage in fractional {B}rownian motion models.
\newblock {\em Finance and {S}tochastics}, {\bf7}(4):533--553, 2003.

\bibitem{Ding}
Z.~Ding, C.~W.~J. Granger, and R.~F. Engle.
\newblock A long memory property of stock market returns and a new model.
\newblock {\em Journal of {E}mpirical {F}inance}, {\bf1}:83--106, 1993.

\bibitem{He}
X.~He and W.~Chen.
\newblock The pricing of credit default swaps under a generalized mixed
  fractional {B}rownian motion.
\newblock {\em Physica {A}: {S}tatistical {M}echanics and its {A}pplications},
  {\bf404}:26--33, 2014.

\bibitem{kemna1990pricing}
Angelien~GZ Kemna and Antonius~CF Vorst.
\newblock A pricing method for options based on average asset values.
\newblock {\em Journal of {B}anking \& {F}inance}, {\bf14}(1):113--129, 1990.

\bibitem{Lo}
A.~W. Lo.
\newblock Long-term memory in stock market prices.
\newblock {\em Econometrica}, {\bf59}:1279--1313, 1991.

\bibitem{MishuraBook}
Y.~S. Mishura.
\newblock {\em Stochastic Calcular for Fractional Brownian Motion and Related
  Process}.
\newblock Springer--{V}erlag, Berlin, 2008.

\bibitem{Mishura}
Y.~S. Mishura and E.~Valkeila.
\newblock {The absence of arbitrage in a mixed {B}rownian fractional {B}rownian
  model}.
\newblock In {\em Proceedings of the {S}teklov {I}nstitute of {M}athematics},
  number 237, pages 224--233. Trudy {M}atematicheskogo {I}nstituta {I}meni
  {V.A.} {S}teklova, 2002.

\bibitem{Peng}
B.~Peng and F.~Peng.
\newblock Pricing {A}sian power options under jump-fraction process.
\newblock {\em Journal of {E}conomics {F}inance and {A}dministrative
  {S}cience}, {\bf17}:2--9, 2012.

\bibitem{Prakasa}
B.~L. S.~Prakasa Rao.
\newblock Pricing geometric {A}sian power options under mixed fractional
  {B}rownian motion environment.
\newblock {\em Physica {A}: {S}tatistical {M}echanics and its {A}pplications},
  {\bf446}:92--99, 2012.

\bibitem{Rostek}
S.~Rostek and R.~Sch\"obel.
\newblock A note on the use of fractional {B}rownian motion for financial
  modeling.
\newblock {\em Economic {M}odelling}, {\bf30}:30--35, 2013.

\bibitem{shokrollahi2018evaluation}
Foad Shokrollahi.
\newblock The evaluation of geometric {A}sian power options under time changed
  mixed fractional {B}rownian motion.
\newblock {\em Journal of {C}omputational and {A}pplied {M}athematics},
  {\bf344}:716--724, 2018.

\bibitem{shokrollahi2014pricing}
Foad Shokrollahi and Adem K{\i}l{\i}{\c{c}}man.
\newblock Pricing currency option in a mixed fractional {B}rownian motion with
  jumps environment.
\newblock {\em Mathematical {P}roblems in {E}ngineering}, {\bf2014}, 2014.

\bibitem{shokrollahi2015actuarial}
Foad Shokrollahi and Adem K{\i}l{\i}{\c{c}}man.
\newblock Actuarial approach in a mixed fractional {B}rownian motion with jumps
  environment for pricing currency option.
\newblock {\em Advances in {D}ifference {E}quations}, {\bf2015}(1):1--8, 2015.

\bibitem{shreve2004stochastic}
Steven~E Shreve.
\newblock {\em Stochastic calculus for finance {II}: {C}ontinuous-time models},
  volume~{\bf11}.
\newblock Springer {S}cience \& {B}usiness {M}edia, 2004.

\bibitem{Sun}
L.~Sun.
\newblock Pricing currency options in the mixed fractional {B}rownian motion.
\newblock {\em Physica {A}: {S}tatistical {M}echanics and its {A}pplications},
  {\bf392}:3441--3458, 2013.

\bibitem{turnbull1991quick}
Stuart~M Turnbull and Lee~Macdonald Wakeman.
\newblock A quick algorithm for pricing {E}uropean average options.
\newblock {\em Journal of {F}inancial and {Q}uantitative {A}nalysis}, pages
  377--389, 1991.

\bibitem{vorst1992prices}
Ton Vorst.
\newblock Prices and hedge ratios of average exchange rate options.
\newblock {\em International {R}eview of {F}inancial {A}nalysis},
  {\bf1}(3):179--193, 1992.

\bibitem{Wang2}
X.-T. Wang, E.-H. Zhu, M.-M. Tang, and H.-G. Yan.
\newblock Scaling and long-range dependence in option pricing {II}: {P}ricing
  {E}uropean option with transaction costs under the mixed
  {B}rownian–fractional {B}rownian model.
\newblock {\em Physica {A}: {S}tatistical {M}echanics and its {A}pplications},
  {\bf3}:445--451, 2010.

\bibitem{Xiao}
W.~L. Xiao, W.~G. Zhang, X.~Zhang, and X.~Zhang.
\newblock Pricing model for equity warrants in a mixed fractional {B}rownian
  environment and its algorithm.
\newblock {\em Physica {A}: {S}tatistical {M}echanics and its {A}pplications},
  {\bf391}:6418--6431, 2012.

\bibitem{ZhangPower}
W.~G. Zhang, Z.~Li, and Y.~J. Liu.
\newblock Analytical pricing of geometric {A}sian power options on an
  underlying driven by a mixed fractional {B}rownian motion.
\newblock {\em Physica {A}: {S}tatistical {M}echanics and its {A}pplications},
  {\bf490}:402--418, 2018.

\bibitem{Zili}
M.~Zili.
\newblock On the mixed fractional {B}rownian motion.
\newblock {\em Journal of {A}pplied {M}athematics and {S}tochastic {A}nalysis},
  {\bf ID 32435}:1--9, 2006.

\end{thebibliography}

\end{document}